\documentclass[11pt,a4paper]{article}

\usepackage{fullpage}

\usepackage{epsfig}
\usepackage{amssymb}

\newtheorem{theorem}{Theorem}

\newenvironment{proof}{{\bf Proof:}}{\hfill\rule{1.5mm}{3mm}\vspace{0.1in}}

\newcommand{\comment}[1]{}

\begin{document}
\title{Low congestion online routing and an improved mistake bound
  for online prediction of graph labeling}

\author{ Jittat Fakcharoenphol\footnote{Department of Computer
    Engineering, Kasetsart University, Bangkok, Thailand
    10900. E-mail: {\tt jittat@gmail.com}. Supported by the Thailand
    Research Fund Grant MRG5080318.} \and Boonserm
  Kijsirikul\footnote{Department of Computer Engineering, Chulalongkorn
    University, Bangkok, Thailand 10330.  E-mail: {\tt
      Boonserm.K@chula.ac.th}.}}

\maketitle\thispagestyle{empty}

\begin{abstract}
In this paper, we show a connection between a certain online
low-congestion routing problem and an online prediction of graph
labeling.  More specifically, we prove that if there exists a routing
scheme that guarantees a congestion of $\alpha$ on any edge, there
exists an online prediction algorithm with mistake bound $\alpha$
times the cut size, which is the size of the cut induced by the label
partitioning of graph vertices.  With previous known bound of $O(\log
n)$ for $\alpha$ for the routing problem on trees with $n$ vertices,
we obtain an improved prediction algorithm for graphs with high
effective resistance.

In contrast to previous approaches that move the graph problem into
problems in vector space using graph Laplacian and rely on the
analysis of the perceptron algorithm, our proof are purely
combinatorial.  Further more, our approach directly generalizes to the
case where labels are not binary.
\end{abstract}

\section{Introduction}
We are interested in an online prediction problem on graphs.  Given a
connected graph $G=(V,E)$ and a labeling $\ell:V\rightarrow\{-1,+1\}$,
unknown to the prediction algorithm, in each round $i$, for
$i=1,2,\ldots$, an adversary asks for a label of a vertex $v_i\in V$,
the prediction algorithm provides the answer $y_i$, and then receives
the correct label $\hat{y}_i=\ell(v_i)$.  The goal is to minimize the
number of rounds that the algorithm makes a mistake, i.e., rounds $i$
such that $y_i\neq\hat{y}_i$.  To make our presentation clean, in this
work we do not count the mistake made on the first question
$v_1$.\footnote{To properly account this, one can simply add 1 to our
  mistake bound.}

This problem has been studied with standard online learning tools such
as the perceptron algorithm.  Herbster, Pontil, and
Wainer~\cite{HerbsterPW-ICML05-Learning-Over-Graphs}, and Herbster and
Pontil~\cite{HerbsterP-NIPS07-Perceptron} use pseudoinverse of graph
Laplacian as a kernel and provide a mistake bound that depends on the
size of the cut induced by the partition based on the real labeling of
vertices and the largest effective resistance between any pair of
vertices in the graph.  Recently, Herbster~\cite{Herbster-ALT08}
exploits the cluster structure of the labeling on the graph, and
provides an improved mistake bounds.

Pelckmans and Suykens~\cite{PelckmansS08-MLG08-labeling} present a
combinatorial algorithm for the problem that predicts a label of a
given vertex based on known labels of its neighbors.  They also prove
a bound on the number of mistakes when the labels of adjacent vertices
are known.  However, their bound is very loose since it does not count
every mistakes and their proof is still based on graph Laplacian.  We
shall compare the bound that we obtain with previous bounds of
Herbster {\em et. al.}~\cite{HerbsterPW-ICML05-Learning-Over-Graphs,
  HerbsterP-NIPS07-Perceptron, Herbster-ALT08} and of Pelckmans and
Suykens~\cite{PelckmansS08-MLG08-labeling} in
Section~\ref{sect:comparison}.

This work follows the initiation of Pelckmans and Suykens.  We show
connection between the prediction problem and the following online
routing problem, first introduced by Awerbuch and
Azar~\cite{AwerbuchA95-multicast} in their study of online multicast
routing.  Given a connected graph $G=(V,E)$, the algorithm receives a
sequence of requests $r_1,r_2,\ldots$, where $r_i\in V$, and, for each
$r_i$, where $i>0$, has to route one unit of flow from $r_i$ to some
previous know $r_j$ where $j<i$.  The algorithm works in an online
fashion, i.e., it has to return a route for $r_i$ before receiving
other requests $r_{i'}$, where $i'>i$.  Given a set of routes, we
define the {\em congestion} $Cong(e)$ incurred on edge $e\in E$,
defined as the number of routes that use $e$.  The performance of the
algorithm is measured by the maximum congestion incurred on any edge.

We prove, in Section~\ref{sect:reduction}, that if there exists an
algorithm $A$ with a guarantee that the congestion incurred on any
edge will be no greater than $\alpha$, there exists an online
prediction algorithm with the mistake bound of 
\[
\alpha\cdot |cut(\ell)|,
\]
where $cut(\ell)$ be the set of edges joining pairs of vertices with
different labels, i.e., $cut(\ell)=\{(u,v)\in E:\ell(u)\neq\ell(v)\}$.

In Section~\ref{sect:mistake-bound}, we apply the known congestion
bound to show the mistake bound for the graph prediction problem, and
compare the bound obtained with the bounds from previous results.

We note that our approach directly generalizes to the case when labels
are not binary (i.e., when the labeling function $\ell$ maps $V$ to an
arbitrary set $L$ of labels) with the same mistake bound.

\section{Reduction to low-congestion routing}
\label{sect:reduction}
We first present an online prediction algorithm from an online routing
algorithm $A$.  The prediction algorithm $P_A$ is very simple, given a
vertex $v_i$, it uses $A$ to route one unit of flow from $v_i$ to any
vertices $v_j$ with known labels, it then returns the known label
$\ell(v_j)$ as the prediction.

We prove the following theorem.

\begin{theorem} 
\label{thm:low-congestion-low-mistake}
If $A$ guarantees that no edges is used more than $\alpha$ times, the
prediction algorithm would make at most $\alpha\cdot |cut(\ell)|$
mistakes, not including the mistake made on the first query $v_1$.
\end{theorem}
\begin{proof}
We shall show that the number of mistake is at most $\alpha\cdot
|cut(\ell)|$.  Note that for each mistake $P_A$ makes on vertex $v_i$,
$A$ routes $v_i$ to some known vertex $v_j$ along a path $P_i$.  Since
$P_A$ predicts $\ell(v_j)$ and makes a mistake, we have
$\ell(v_i)\neq\ell(v_j)$; thus, $P_i$ must use some cut edge $e$ in
$cut(\ell)$.  We charge this mistake to $e$.  We note that $P_i$ may
use many cut edges, but we only charge the mistake to one arbitrary
edge.  Since the routing produced by $A$ uses each edge no more than
$\alpha$ times, each cut edge is charged no more than $\alpha$ times
as well.  Therefore, the number of mistakes $P_A$ makes must be at
most $\alpha\cdot|cut(\ell)|$, as required.
\end{proof}

We note that this proof does not use any fact that the labeling $\ell$
is binary; therefore, the proof holds for general labeling as well.

\section{Mistake bound}
\label{sect:mistake-bound}

To obtain the mistake bound, we first state the result on the online
routing on trees.  The theorem below first appeared in the work of
Awerbuch and Azar~\cite{AwerbuchA95-multicast}, in which they called
the problem {\em restricted offline multicast}, and has been
discovered independently by Chalermsook and
Fakcharoenphol~\cite{ChalermsookF05}.  We state the result in the form
in~\cite{ChalermsookF05} as it matches our settings.

\begin{theorem}
[Theorem 4.4 in~\cite{AwerbuchA95-multicast}, Theorem 1
  in~\cite{ChalermsookF05}]
\label{thm:low-congestion-routing}
For any tree $T$ with $n$ vertices and any sequence of vertices
$t_1,t_2,\ldots t_k$ in $T$, there exists an efficient algorithm that
finds a set of paths $q_1, q_2,\ldots, q_{k-1}$ such that (1) $q_i$
connects $t_{i+1}$ to some $t_j$, such that $j\leq i$, and (2) each
edge in $T$ belongs to at most $O(\log n)$ paths.  Moreover the path
$q_i$ depends only on paths $q_1,q_2,\ldots,q_{i-1}$.
\end{theorem}

We note that the bound also holds for general graph $G$ by taking $T$
to be its spanning tree.

Using Theorems~\ref{thm:low-congestion-low-mistake}
and~\ref{thm:low-congestion-routing}, we obtain the following mistake
bound.

\begin{theorem}
For graph $G=(V,E)$ and an unknown labeling $\ell:V\rightarrow L$,
there exists an efficient prediction algorithm that makes at most
\[
O(\log |V|)\cdot|cut(\ell)|
\]
mistakes, where $cut(\ell)$ denotes the set of edges joining pairs of
vertices with different labels.
\end{theorem}

We note that for line graph, our algorithm is optimal.  One can prove,
in the same way as the proof of optimality of binary search, that an
adversary can fool any algorithm to make $\Omega(\log n)$ mistakes on
a line.

\subsection{Comparison to previous bounds}
\label{sect:comparison}

We compare our mistake bound with the previous results.

\begin{itemize}
\item 
Herbster {\em et. al.}~\cite{HerbsterPW-ICML05-Learning-Over-Graphs,
  HerbsterP-NIPS07-Perceptron} present an algorithm based on
perceptron and prove the bound of
\[
4\cdot |cut(\ell)|\cdot R_G + 2,
\]
for the number of mistakes where $R_G$ is the largest effective
resistance between any pair of nodes in $G$
(see~\cite{HerbsterP-NIPS07-Perceptron}, for the formal definition).
We note that there are graphs where $R_G$ is large, e.g, for line
graph $R_G=n-1$.  Our bound is better when $R_G=\Omega(\log n)$.

While in the worst case $R_G$ can be large, for many classes of
graphs, e.g., highly connected graphs with small diameter, $R_G$ can
be very small.  In~\cite{HerbsterP-NIPS07-Perceptron}, they give an
example where the cut size $|cut(\ell)|$ is linear, while $R_G$ is
$O(1/|cut(\ell)|)$.  In this example, their mistake bound remains
constant, while our bound grows with $|cut(\ell)|$.

\item
In a recent paper, Herbster~\cite{Herbster-ALT08} exploits the cluster
structures of graphs and proves the bound of
\[
{\mathcal N}(G,\rho) + 4\cdot|cut(\ell)|\cdot\rho +1
\]
for any $\rho>0$, on the number of mistakes, where ${\mathcal
  N}(X,\rho)$, the covering number, is the minimum number of sets of
diameter $\rho$ that contain all vertices of $G$ under the semi-norm
induced by the graph Laplacian (see~\cite{Herbster-ALT08} for
definitions).

This bound improves over previous bound
in~\cite{HerbsterP-NIPS07-Perceptron} when the graph has small number
of clusters with small diameters.  Herbster gives an example where the
new algorithm makes only a constant number of mistakes while the
algorithm from~\cite{HerbsterP-NIPS07-Perceptron} makes linear
mistakes.  Again, in this example, our algorithm has linear mistake
bound.

We note that there is a trade-off between the diameter $\rho$ of
clusters and the number clusters in Herbster's bound.  For many
classes of graphs with large diameter, e.g. line graphs, using cluster
structure does not help.  The dependent on the cut size can still be
$\Omega(n)$ for graphs with $n$ vertices.

\item
Pelckmans and Suykens~\cite{PelckmansS08-MLG08-labeling} present a
simple combinatorial algorithm and show that the set $M$ of vertices
where the algorithm predicts incorrectly satisfies $\sum_{v\in M}
d_{M,v}\leq 4\cdot |cut(\ell)|$, where $d_{M,v}$ is the number of
vertices adjacent to $v$ that is also in $M$.  Note that their bound
only accounts for edges between two mistaken vertices.  If there are
no edges between vertices in $M$, their bound does not say anything.
For example, consider the case with line graph with $n$ vertices,
where vertices $1,2,\ldots,n/2$ have label $+1$ and vertices
$n/2+1,\ldots,n$ have label $-1$.  The algorithm of Pelckmans and
Suykens can make $\Omega(n)$ mistakes if an adversary asks the labels
of $1,3,5,\ldots$, while the cut size is just 1.
\end{itemize}

\section{Open questions and discussions}

Our bound depends on the worst case bound on the congestion from the
routing problem.  However, the $O(\log n)$ bound seems very loose for
dense graphs.  It would be nice to see if one can find the connection
between the worst case congestion and the effective resistance.  We
note that when the effective resistance is low, between any two nodes
there must be many short disjoint paths, and this should help reducing
the congestion.  Also, there is extensive literature on online routing
with small congestion (see,
e.g.,~\cite{Racke-FOCS02,HarrelsonHR-SPAA03,Racke-STOC08}).  Can these
results be used to give better mistake bounds as well?

We note that our proof cannot give a mistake bound smaller than
$|cut(\ell)|$.  To improve further, one need a way to account for cut
edges that have not been charged.

Finally, we wish to see any adversarial bound on the number of
mistakes for an online label prediction algorithm.  In this paper, we
have shown that our algorithm is optimal (up to a constant factor) for
line graphs.  The ultimate goal would be to find an optimal algorithm
for general graphs.

\bibliographystyle{plain} \bibliography{predict}

\end{document}